\documentclass[aps,twocolumn,preprintnumbers]{revtex4}

\usepackage{amsthm}
\usepackage[breaklinks=true,colorlinks,citecolor=blue,linkcolor=blue,urlcolor=blue]{hyperref}
\usepackage[mathlines]{lineno}
\usepackage[caption=false, position=bottom]{subfig}
\usepackage{floatrow}
\usepackage{mathtools}
\usepackage{graphicx}
\usepackage{xcolor}
\usepackage{fancyhdr}
\usepackage{multirow}
\usepackage{fancyhdr}
\usepackage{longtable}
\usepackage[T1]{fontenc}
\usepackage{dcolumn}   

\usepackage{bm}        
\usepackage{amsfonts}  
\usepackage{amsmath}   
\usepackage{amssymb}   

\newtheorem{lemma}{Lemma}

\newtheorem{remark}{Remark}

\usepackage[utf8]{inputenc}
\usepackage[english]{babel}
\usepackage{physics}

\newtheorem{theorem}{Theorem}

\usepackage{amsmath,amsfonts}
\DeclareMathOperator{\spn}{span}

\usepackage{amsmath}
\usepackage{amssymb}

\begin{document}

\title{Generalized Fusion of Qudit Graph States}

\author{Noam Rimock and 
Yaron Oz}

\affiliation{School of Physics and Astronomy, Tel Aviv University, Ramat Aviv 69978, Israel}


\begin{abstract}

We formalize a generalized type-II fusion operation for qudit cluster states within linear optics. Two designated qudits, one from each input cluster, interfere with optional ancilla qudits via a passive linear-optical network, followed by number-resolving detection; conditioned on 
measurement outcome,
the remaining qudits form the post-selected fused state. We prove a general rank bound: for any such interferometer and outcome, the reduced density matrix across the two parent clusters has Schmidt rank at most $M$, the total number of measured qudits including ancillae. Consequently, a correct qudit fusion which requires rank $d$ is impossible without ancillae and requires at least $d-2$ ancilla qudits. Our analysis extends previous no-go results for Bell-type qubit fusion to the qudit setting and to generalized, non-Bell projections. We analyze the probabilities and entanglement of the relevant measurement outcomes, and discuss how our lower bound aligns with existing constructive schemes. These results set a clear resource threshold for high-dimensional, fusion-based photonic MBQC.

\if{

In qudit MBQC,  the resource state is a qudit graph state built with qudits as nodes,
controlled-phase gates as edges, and single‐site measurements are in the $d$-dimensional Pauli  algebra.
It has been proven that type II fusion applied to qudits graph states via Bell-type measurement of qudits does not form the required fused graph, and there is an upper bound on the rank of reduced density matrix by the number of measured qudits including the ancillae. 
We define generalized type II fusion of qudit graphs states, 
and prove that it requires at least $d-2$ ancillae.

}\fi
\end{abstract}


\maketitle

\section{Introduction}

Measurement based quantum computing (MBQC) proceeds by preparing a highly entangled resource state and then driving the computation with single‑site measurements. In photonics, large‑scale cluster states are assembled by fusion operations that interfere qubits or qubits from smaller resource states and post‑select on specific detector outcomes. For qubits, linear‑optical type‑I/II fusion  \cite{BrowneRudolph,FusionBasedQC,ThreePhoton} has become a standard primitive.
Generalized fusion has been introduced in \cite{bartolucci2021creationentangledphotonicstates,rimock2024generalized,schmidt2024generalizedfusionsphotonicquantum}, where the projection is not onto a Bell state, but rather on a general maximally-entangled state, resulting in the required fused state up to local qubit rotations.

Extending these fusion tools to qudits, i.e. higher‑dimensional photonic modes, remains a largely open problem. Qudits promise denser information encoding, improved noise tolerance, and hardware‑efficient entanglement distribution, yet fusing qudit clusters with only passive optics and photon counting requires further study.

This work formalizes a generalized, linear‑optical type‑II fusion for qudit cluster states and establishes a fundamental resource bound. In our setting, one qudit from each input cluster, the fused legs, is interfered on a passive, number‑preserving optical network together with optional ancilla qudits; number‑resolving detection on these measured modes heralds a measurement outcome, leaving a post‑selected state on the remaining, unmeasured cluster qudits. Our aim is to investigate when can such an operation realize the ideal qudit cluster fusion, up to local unitaries. We will 
obtain a simple rank constraint that translates directly into an ancilla requirement.

\if{
Measurement‐based quantum computation (MBQC) is a paradigm that uses a highly entangled resource state in the form of a graph state, and the computation proceeds by performing single‐qubit measurements of the
resource state with a classical feed‐forward. The graph state is multi‐qubit entangled state is
defined directly by an undirected graph $G=(V,E)$, where the vertices are singl-qubit states
and the edges are two-qubit entangling gates. 
Large graph states are typically built by combining smaller entangled fragments via fusion operations, which are probabilistic two‐qubit measurements. 
Such a fusion is the type-II fusion proposed in \cite{BrowneRudolph}, which is based on Bell states measurement \cite{FusionBasedQC,ThreePhoton}.

Generalized fusion has been introduced in \cite{bartolucci2021creationentangledphotonicstates,rimock2024generalized,schmidt2024generalizedfusionsphotonicquantum}, where the projection is not onto a Bell state, but rather on a general maximally-entangled state, resulting in the required fused state up to local qubit rotations.
The aim of this paper is to define and study generalized fusion for qudits,
i.e. a $d$-level quantum system. We will refer to qudits as photons, which allow a realization of qudits as spatial modes.
}\fi

In qudit MBQC,  the resource state is a qudit graph state built with qudits as nodes
and controlled-phase gates as edges, and single‐site measurements are in the $d$-dimensional Pauli  algebra
\cite{Zhou_2003}. It has been proven in \cite{Calsamiglia_2002}, that a straightforward generalization of Bell measurement to qudits is not applicable, as the reduced density matrix when partitioning the fused state to the two original $d$-graphs is at most of rank 2, even when adding vacuum ancillae. Furthermore, it has been proven that when using excited ancillae, the rank of reduced density matrix is bounded above by the number of qudits that are measured, including the ancillae. 
In this work we will generalize this proof for the generalized fusion process, in theorems \ref{Theorem:rank(rho) is less than 2},\ref{Theorem:rank is less or equal to M} and \ref{Theorem:rank is less or equal to M'}. Thus, for qudits one has to use at least $d-2$ ancillae, even for the generalized fusion. A scheme for that has been proposed in \cite{Luo_2019} using $d$ Bell pairs.

In the paper we refer to qudit for a $d$‑dimensional photonic mode, where measured qudits include both fused legs and any ancillae routed to detectors, and relevant outcomes are the multi‑click patterns intended to herald fusion. All our results assume passive, number‑preserving linear optics and number‑resolving detection; no feed‑forward or nonlinear interactions are required for the bounds to hold.

The paper is organized as follows: in section \ref{sec:Qudits Cluster} we review the concepts of qudits and qudit cluster state \ref{subsec:cluster state of qudits}. 
In section \ref{sec:Generalized fusion of qudits clusters} we begin by reviewing 
type II fusion of qubit clusters \ref{subsec:fusion type-II}. We then introduce our definition for the generalized fusion type-II of clusters of qudits, and perform an analytical analysis \ref{subsec:The d-bits based type 2 fusion process}. We prove that this fusion process can't produce the required fused cluster state of qudits \ref{subsec:Achieving maximally entangled final state is impossible unless}. Lastly, we add ancilla resources to the generalized fusion, carry out the analytical analysis and prove that at least $d-2$ ancilla qudits are required for a successful fusion, first when assuming the ancilla state is a product state of all the ancilla channels \ref{subsec:Measuring more than two qudits}, then for the general ancilla state \ref{subsec:Measuring more than two qudits with Entangled Ancilla Qudits}.  Section \ref{sec:discussion and outlook} is devoted to a discussion and outlook.

\section{Qudit Cluster States}
\label{sec:Qudits Cluster}
In this section briefly review the qudits, Pauli operators and their cluster states.

\subsection{Qudits}
\label{subsec:Qudits}

A qudit is a $d$-level quantum system, described by a $d$-dimensional Hilbert space $\mathcal{H}_d$ with a canonical basis
$\bigl\{\lvert0\rangle,\lvert1\rangle,\dots,\lvert d-1\rangle\bigr\}$.
A general qudit state takes the form:
\begin{equation}
    \ket{\phi}=\sum_{i=0}^{i=d-1} \alpha_i \ket{i} \ ,
\end{equation}
where $\Sigma_{i=0}^{i=d-1} \abs{\alpha_i}^2 =1$.
One defines a generalization of the Pauli operators that act on single qudits \cite{Wang_2020}:
\begin{eqnarray}
Z |j\rangle &=& \omega^j |j\rangle, \nonumber\\
X |j\rangle &=& |j+1 \mod d\rangle \ ,
\end{eqnarray}
where $\omega = e^{2\pi i/d}$ is the $d$th root of unity.
They satisfy $XZ = \omega ZX$. Explicitly,
\begin{eqnarray}
X=\sum_{j=0}^{d-1}\ket{j+1}\bra{j},\qquad
Z=\sum_{j=0}^{d-1}\omega^j\ket{j}\bra{j} \ .
\end{eqnarray}

\subsection{Qudit Cluster States}
\label{subsec:cluster state of qudits}
Given a graph, we associate with it a graph state which is an entangled 
quantum state of qubits:
\begin{gather}
    \ket{\phi}_C=\prod_{(a,b)\in E}\left(\ket{0}_a\bra{0}_a ç_b + \ket{1}_a\bra{1}_a Z_b\right) \prod_{c\in V} \ket{+}_c \ ,
\end{gather}
where $V$ is the set of vertices, $E$ the set of edges, $\ket{+}_c=\frac{\ket{0}_c+\ket{1}_c}{\sqrt{2}}$ is 
the qubit state, and $\mathbb{I}$ is the identity operator. This quantum state has an equivalent definition as the solution to a set of linear equations:
\begin{gather}
    X_a \prod_{(a,b)\in E} Z_b \ket{\phi}_C = \ket{\phi}_C \ ,
\end{gather}
for every $a\in V$. The operators $K_a=X_a \prod_{(a,b)\in E} Z_b$ are called stabilizers. 
Two-dimensional cluster states are graph states that have the structure of a two-dimensional grid, and have been
proven to provide a resource state for a universal MBQC quantum computation.
One can generalize the qubit cluster state to a qudit 
cluster as follows \cite{Zhou_2003}: 
\begin{gather}
    \ket{\phi}_C=\prod_{(a,b)\in E}S_{ab} \prod_{c\in V} \ket{+}_c \ ,
\end{gather}
where the operators $S_{ab}$ are:
\begin{gather}
    S_{ab}=\sum_{j,k=0}^{j,k=d-1} \omega^{jk}\ket{j}_a\ket{k}_b\bra{j}_a\bra{k}_b \ ,
\end{gather}
and the initial states of the qudits are:
\begin{gather}
    \ket{+}_c=\frac{1}{\sqrt{d}}\sum_{j=0}^{j=d-1} \ket{j} \ .
\end{gather}

This state has an equivalent definition as a solution  to a set of linear equations:
\begin{gather}
    K_a \ket{\phi}_C=\ket{\phi}_C \ , \nonumber \\
    K_a=X_{a}^{\dag}\prod_{(a.b)\in E} Z_{b} \ ,
\end{gather}
for every $a\in V$. Note, that for $d=2$ the Pauli operator $X$ is hermitian, $X^\dag =X$,  and this reduces correctly to the qubits cluster state definition.
As with qubits, the two-dimensional qudit cluster state is a resource state for a universal quantum computation \cite{Zhou_2003}.

\section{Generalized fusion of qudits clusters}
\label{sec:Generalized fusion of qudits clusters}

\subsection{Qubit Generalized  Type-II Fusion}
\label{subsec:fusion type-II}

The construction of resource states such as cluster states is fundamental to MBQC. Moreover,
generation of cluster states is required during computation since measurements erase parts of the cluster. Fusion gates allow the construction of large cluster states from smaller ones.
These gates are probabilistic: with a certain probability the fusion of the two clusters succeeds.
Otherwise the fusion fails and the clusters remain separated. Two basic fusion gates are the type-I and type-II fusions, proposed in \cite{Browne_2005}. Fusion type-I is used to construct one-dimensional clusters, and fusion type-II is employed to fuse two one-dimensional clusters to "+" shape, as in figure 4 in \cite{rimock2024generalized}.

Mathematically, the fusion type-II procedure is as follows. One begins
by choosing qubits $a$ and $b$ from the two one-dimensional clusters (where in fact $a$ is a part of a logical qubit consisting of $a$ and another qubit $e$). Each qubit is realized by horizontal and vertical polarizations of a certain frequency/location channel, where the photon being in the horizontal state, $a_H^\dag \ket{vac}$, is the computational state $\ket{0}$ and the photon being in the vertical state, $a_V^\dag \ket{vac}$, is the computational state $\ket{1}$. Then the creation operators associated with the channels of qubits $a$ and $b$ are transformed by a diagonal polarization beam splitter into which the photons of $a,b$ channels are directed.
\begin{eqnarray}
    \begin{bmatrix}
        a^\dag_H \\
        a^\dag_V \\
        b^\dag_H \\
        b^\dag_V
    \end{bmatrix} &=& \begin{bmatrix}
     U_{11} & U_{12} & U_{13} & U_{14} \\
     U_{21} & U_{22} & U_{23} & U_{24} \\
     U_{31} & U_{32} & U_{33} & U_{34} \\
     U_{41} & U_{42} & U_{43} & U_{44}
    \end{bmatrix}
    \begin{bmatrix}
        c^\dag_H \\
        c^\dag_V \\
        d^\dag_H \\
        d^\dag_V
    \end{bmatrix}, \nonumber\\
    U &=&\frac{1}{2}\begin{bmatrix}
    1 & 1 & 1 & -1 \\
     1 & 1 & -1 & 1 \\
     1 & -1 & 1 & 1 \\
     -1 & 1 & 1 & 1
\end{bmatrix} \ .
\label{eq:UnitaryTransformationForRegularFusionType2}
\end{eqnarray}

The next step is to measure the new channels $c_H^\dag,c_V^\dag,d_H^\dag,d_V^\dag$ that come out of the beam splitter. 
There is a $0.5$ probability of measuring one photon in two different channels out of the four possible channels, 
which is the result of applying a projection operator on one of the Bell states of $a,b$:
\begin{gather}
    \frac{1}{\sqrt{2}}\left(\ket{0}_a\ket{0}_b\pm \ket{1}_a\ket{1}_b\right), \frac{1}{\sqrt{2}}\left(\ket{0}_a\ket{1}_b\pm \ket{1}_a\ket{0}_b\right) \ .
\end{gather}
Applying this projection operator to the two cluster states results in the required fused cluster state. 
The other case, where one measures two photons in one of the four channels, has also $0.5$ probability and results in a failure of the fusion - the two clusters remain separated, and the full quantum state is the product of their appropriated quantum states.

Thus, the essence of fusion type-II is projecting onto Bell states, which is analogous to Bell-measurements. Using linear optical components, one can employ a more general unitary $U$ matrix in (\ref{eq:UnitaryTransformationForRegularFusionType2}) \cite{CreatingU0}, with the same aim getting a Bell state projection with as large probability of success as possible. The maximal probability has been proven to be indeed $0.5$ when using only linear elements \cite{Calsamiglia_2001}, and can be increased by adding ancilla qubits \cite{grice2011arbitrarily,Ewert_2014,Kilmer_2019,Bayerbach_2023}. In \cite{bartolucci2021creationentangledphotonicstates} a protocol involving projections onto states of $a,b$ that are not Bell states was introduced. This protocol used ancilla qubits, and surpassed the maximal probability of the standard type-II fusion (using the same number of ancilla qubits). In \cite{rimock2024generalized,schmidt2024generalizedfusionsphotonicquantum} the generalized fusion protocol of projecting on a general state of $a,b$ was further studied. It has been proven that without using ancilla qubits, the generalized fusion has the same maximal probability of success, $0.5$, as the standard one \cite{rimock2024generalized,Calsamiglia_2001}. Also, it has been shown numerically that by adding ancilla qubits, the probability of success of the generalized fusion is higher than that of the regular fusion with same number of ancilla qubits \cite{schmidt2024generalizedfusionsphotonicquantum}. In both cases the aim of the generalized fusion is to project onto a maximally-entangled state of $a,b$.

For our further analysis in this paper, we approach generalized fusion type-II as follows. We denote by $X_1,X_2$ the Hilbert spaces that are the tensor products of the Hilbert spaces of the qubits of the first and second original cluster, respectively. We also denote by $W_1$,$W_2$ the Hilbert spaces of the qubits $a$ and $b$, respectively, i.e $W_1=span\left({a_H^\dag \ket{vac},a_V^\dag \ket{vac}}\right)$ and $W_2=span\left({b_H^\dag \ket{vac},b_V^\dag \ket{vac}}\right)$. That means one can write $X_1=V_1\otimes W_1$ and $X_2=V_2\otimes W_2$, where $V_1,V_2$ are the tensor products of the Hilbert spaces of the qubits of the first and second original cluster without the qubits $a$ and $b$, respectively. Then, using linear elements is equivalent to applying some unitary on $span\left({a_H,a_V,b_H,b_V}\right)$ that transfers between the bases $a_H,a_V,b_H,b_V$ and $c_H,c_V,d_H,d_V$. In fact, $c_H,c_V,d_H,d_V$ will not be in $span\left({a_H,a_V,b_H,b_V}\right)$ but in a same dimension Hilbert space, but we can identify it with $span\left({a_H,a_V,b_H,b_V}\right)$. Finally, one measures the channels $c_H,c_V,d_H,d_V$, and as a result is left with a wave function in the Hilbert space $V_1 \otimes V_2$, which we desire to be a cluster state of the qubits of the two original clusters except for $a,b$. If one is using $N$ vacuum ancilla qubits, then one does an isometry from $span\left({a_H,a_V,b_H,b_V}\right)$ to $span\left({c_H,c_V,d_H,d_V,c_1,...,c_N}\right)$, which is equivalent to a unitary from $span\left({a_H,a_V,b_H,b_V,c_1,...,c_N}\right)$ to itself. Notice that $a_H,a_V$ being two polarizations of the same channel (and the same for $b_H,b_V$) is not important for the mathematical structure of the process, only for its realization; so one can also think of the channels $a_H,a_V,b_H,b_V$ simply as four different channels.

This analysis can also be done for qudits. Protocols for fusion type-II of qudits were introduced in \cite{Luo_2019,bharoshigh,Bharos_2025,_st_n_2025}. It was proven that a Bell-measurement of qudits requires at least $d-2$ ancilla qudits \cite{Calsamiglia_2002}, which means that the same minimal number of ancilla qudits is required for fusion type-II. In the following, we will analyze generalized type-II fusion of qudits, and show that the same requirement of at least $d-2$ ancilla qudits still holds.

\subsection{Qudit Generalized Type-II Fusion}

Here, we consider a generalized case of Fusion type II, where we have two qudits clusters with quantum states $\ket{\Phi_1}$ and $\ket{\Phi_2}$ in Hilbert spaces $X_1$ and $X_2$, respectively. We choose from each cluster a single qudit, $e_1$ and $e_2$, and denote their Hilbert spaces as $W_1$ and $W_2$. Thus, we write $X_i=V_i\otimes W_i$ when $V_i$ is the Hilbert spaces of all the qudits in the cluster $i$ that differ from $e_i$. Every qudit is realized by $d$ different modes, each one representing one of the states $\ket{0},...,\ket{d-1}$ when there is one photon occupying it, and in relation to qubits we can think of every qudit as some photon spatial/frequency mode with $d$ polarizations, hence $d$ appropriate creation operators.

The fusion process consists of:
i) Performing isometry transformation from the space $W_1\oplus W_2$ to some bigger Hilbert space. We can assume, without loss of generality, that the bigger Hilbert space contains $W_1\oplus W_2$. We can complete this to a unitary transformation from $W_1\oplus W_2\oplus W_3$ to itself, where $W_3$ is some Hilbert space that represents vacuum ancilla qudits. 
ii) Measuring the numbers of photons in all the channels of the new basis of $W_1\oplus W_2\oplus W_3$ that we got after the transformation. We will get from the measurement two photons, either in two different modes (which we will refer to as relevant states) or in the same mode (which we will refer to as non-relevant states). Here by 'mode' we refer to a specific 1-dim Hilbert space, not a mode with some polarizations.
For each possible outcome of our measurement, the resulting state is in $V_1\otimes V_2$, which we aim to be a qudit cluster state (up to single-qudit rotations). Each outcome has its probability, and our goal is to maximize the probability to get qudit cluster state. But, as proven below, when $d>2$ all the possible final states are necessarily not qudit cluster states, thus the generalization of fusion type II to qudits is not possible. 

In the next subsection \ref{subsec:The d-bits based type 2 fusion process}, we begin by an analytical analysis of the possible resulting states of the fusion, and the probability and entanglement entropy of each one of them.  This may seem odd, because in subsection \ref{subsec:Achieving maximally entangled final state is impossible unless} we prove that this fusion can't work - the resulting state will necessarily not be the required fused cluster state. But, some of those results are more general by those we achieved in the analytical analysis in \cite{rimock2024generalized}, and more importantly, we believe that subsections \ref{subsec:The d-bits based type 2 fusion process} and \ref{subsec:Achieving maximally entangled final state is impossible unless} are helpful for easily understanding the generalization arguments by adding ancilla qudits in subsection \ref{subsec:Measuring more than two qudits}. In subsection \ref{subsec:Measuring more than two qudits} we perform the analytical analysis of the generalized type-II fusion with ancilla qudits, and prove that at least $d-2$ ancilla qudits need to be used in order to get a successful fusion.

\subsection{Type II Fusion of $D$-Qudit Graph States}
\label{subsec:The d-bits based type 2 fusion process}
Consider a general data, $X_i,V_i,W_i,\ket{\Phi_i}, i=1,2$. In the following
we will not require a maximal entanglement entropy of the fused
state, where the entanglement entropy is computed by the partition of $V_1\otimes V_2$ to $V_1$ and $V_2$.
Denote $n_i=\dim{V_i}$ and $m_i=\dim{W_i}$, where $m_i=d$ and  $k_i=min(n_i,m_i)$ for a qudit cluster state. By the Schmidt decomposition there are orthonormal bases $\{\ket{\phi_{1,i}}\}_{i=1}^{n_1}$, $\{\ket{\phi_{2,i}}\}_{i=1}^{n_2}$, $\{\ket{\psi_{1,i}}\}_{i=1}^{m_1}$ and $\{\ket{\psi_{2,i}}\}_{i=1}^{m_2}$ for $V_1$, $V_2$, $W_1$ and $W_2$, respectively, and non-negative real numbers $\{\alpha_i\}_{i=1}^{k_1}$ and $\{\beta_i\}_{i=1}^{k_2}$, such that the quantum states of the two clusters before the fusion are:
\begin{gather}
    \ket{\Phi_1}=\sum_{i=1}^{k_1} \alpha_{i} \ket{\phi_{1,i}} \ket{\psi_{1,i}},
    \nonumber \\
    \ket{\Phi_2}=\sum_{j=1}^{k_2} \beta_{j} \ket{\phi_{2,j}} \ket{\psi_{2,j}} \ .
    \label{eq:The original wave functions of the original clusters}
\end{gather}
The quantum state of the two clusters before the fusion
is the product state:
\begin{gather}
  \ket{\Phi}=
  \underbrace{\left( {\sum_{i=1}^{k_1} \alpha_{i} \ket{\phi_{1,i}} \ket{\psi_{1,i}}}\right)}_{|\Phi_1\rangle} \underbrace{\left(\sum_{j=1}^{k_2} \beta_{j} \ket{\phi_{2,j}} \ket{\psi_{2,j}}\right)}_{|\Phi_2\rangle} \ .
    \label{eq:Total wave function before changing bases and fusion}
\end{gather}
If $\ket{\Phi_i}$ correspond to $D$-dimensional qudit clusters, the entanglement entropy when 
partitioning $X_i$ to $V_i$,$W_i$ is maximal, hence the $\alpha$'s are equal to $\frac{1}{\sqrt{k_1}}$, the $\beta$'s are equal to $\frac{1}{\sqrt{k_2}}$, and  $k_1=k_2=D$.

The generalizes type II fusion is performed by a unitary transformation from $W_1\oplus W_2 \oplus W_3$ to itself, where $W_3$ is the vacuum ancilla Hilbert space. This transformation is equivalent to a unitary transformation on the creation operators of the modes, which will be denoted by $U^{\dag}$. 
We can associate every quantum state $\ket{\phi_{1,i}}$, $\ket{\phi_{2,i}}$ with creation operator $a_i^{\dag}$, $b_i^{\dag}$ respectively. Those operators are the channels before operating with $U$.  
After operating with $U^\dag$ we replace $\{a_i^{\dag}\}_{i=1}^{i=k_1}$ and $\{b_i^{\dag}\}_{i=1}^{i=k_2}$ with new channels $\{c_i\}_{i=1}^{i=k_1+k_2}$ such that:
\begin{gather}
    \begin{bmatrix}
        a_1^{\dag} \\
        \vdots \\
        a_{k_1}^{\dag} \\
        b_1^{\dag} \\
        \vdots \\
        b_{k_2}^{\dag} \\
        vac_1^{\dag} \\
        \vdots \\
        vac_{k_3}^{\dag} 
    \end{bmatrix}
    = U
    \begin{bmatrix}
        c_1 \\
        \vdots \\
        c_{K} 
    \end{bmatrix} \ .
    \label{eq:Fusion matrix U}
\end{gather}
when $K=k_1+k_2+k_3$ and the quantum state (\ref{eq:Total wave function before changing bases and fusion}) takes the form:
\begin{align}
    \ket{\Phi} & = \left(\sum_{i=1}^{k_1}\sum_{k=1}^{K} U_{i,k}\alpha_i \ket{\phi_{1,i}}c_k^{\dag}\right) \nonumber \\
    &  \cdot \left(\sum_{j=k_1+1}^{K}\sum_{l=1}^{k_1+k_2}U_{j,l}\alpha_j \ket{\phi_{2,j}}c_l^{\dag}\right) \ .
\end{align}


Denote:
\begin{gather}
    a_{ij,kl}=U_{ik}U_{jl}+U_{il}U_{jk}
    \label{eq:aijkl of relevant state kl for general alpha and beta} \ .
\end{gather}
Then the quantum state $\ket{\Phi}$ can be written as:
\begin{eqnarray}
&\sum_{1 \leq k < l \leq K} \sum_{i=1, j=k_1+1}^{i=k_1, j=k_1+k_2} \alpha_i \beta_j a_{ij,kl} \ket{\phi_{1,i}} \ket{\phi_{2,j}} c_k^{\dag} c_l^{\dag} 
    \nonumber \\
    &+\frac{1}{2}\sum_{1 \leq k \leq K} \sum_{i=1, j=k_1+1}^{i=k_1, j=k_1+k_2} \alpha_i \beta_j a_{ij,kk} \ket{\phi_{1,i}} \ket{\phi_{2,j}}c_k^{\dag}c_k^{\dag}, \nonumber
    \end{eqnarray}
    and can be recast as:
    \begin{eqnarray}
   &\sum_{1 \leq k < l \leq K} \sum_{i=1, j=k_1+1}^{i=k_1, j=k_1+k_2} \alpha_i \beta_j a_{ij,kl} \ket{\phi_{1,i}} \ket{\phi_{2,j}} \ket{1}_k \ket{1}_l + \nonumber \\
   & +\frac{\sqrt{2}}{2}\sum_{1 \leq k \leq K} \sum_{i=1, j=k_1+1}^{i=k_1, j=k_1+k_2} \alpha_i \beta_j a_{ij,kk} \ket{\phi_{1,i}} \ket{\phi_{2,j}}\ket{2}_k .
\end{eqnarray}

Next, we measure all the $c_i$ channels. Given that we measured a qudit in channel $k$ and a qudit in channel $l$ ($k$ and $l$ can be the same channel), the resulting quantum state is:
\begin{gather}
    \label{eq:wave function kl for general alpha and beta}
    \ket{\Phi}_{kl}=\frac{1}{N_{kl}}\sum_{i=1, j=k_1+1}^{i=k_1, j=k_1+k_2} \alpha_i \beta_j a_{ij,kl} \ket{\phi_{1,i}} \ket{\phi_{2,j}} \ ,
\end{gather}
where $N_{kl}$ is a normalization factor:
\begin{gather}
    N_{kl}=\sqrt{\sum_{i=1, j=k_1+1}^{i=k_1, j=k_1+k_2} \abs{\alpha_i \beta_j a_{ij,kl}}^2} 
    \label{eq:normalization factor of relevant state kl for general alpha and beta} \ .
\end{gather}

For the simplicity of notations, in the following we will refer to the state (\ref{eq:wave function kl for general alpha and beta}) as a $(k,l)$ state. Note, that when $k=l$ one can write the state (\ref{eq:wave function kl for general alpha and beta}) as:
\begin{gather}
    \ket{\Phi}_{kk}=\frac{2}{N_{kk}}\left(\sum_{i=1}^{i=k_1}  \alpha_i U_{ik}\ket{\phi_{1,i}}  \right) \left( \sum_{j=k_1+1}^{j=k_1+k_2}  \beta_i U_{jk}\ket{\phi_{2,j}} \right) \ ,
    \label{eq:wave function of non-relevant state kk for general alpha and beta}
\end{gather}
which means that it is a product state. Hence, in the following we will refer to the states $(k,k)$ as {\it irrelevant states}, and for the other states $(k,l)$ with $k\ne l$, which  are potentially entangled, we shall refer to as {\it relevant states}.

The probability to get the relevant state $\ket{\Phi}_{kl}$ is:
\begin{gather}
    p_{kl}=N_{kl}^2=\sum_{i=1, j=k_1+1}^{i=k_1, j=k_1+k_2} \abs{\alpha_i \beta_j a_{ij,kl}}^2 \ ,
    \label{eq:probability of relevant state kl for general alpha and beta}
\end{gather}
and the probability to get the irrelevant state $\ket{\Phi}_{kk}$ is:
\begin{gather}
    p_{kk}=\frac{N_{kk}^2}{2}=2 \left(\sum_{i=1}^{i=k_1}  \abs{\alpha_i U_{ik}}^2  \right) \left( \sum_{j=k_1+1}^{j=k_1+k_2}  \abs{\beta_j U_{jk}}^2 \right) \ .
    \label{eq:probability of non-relevant state kk for general alpha and beta}
\end{gather}

The reduced density matrix $\rho_{kl}$ which is obtained by tracing out $\ket{\phi_{2,j}}$ from $\ket{\Phi}_{kl}\bra{\Phi}_{kl}$ reads:
\begin{gather}
    \rho_{kl}=\sum_{1 \leq i_1,i_2 \leq k_1} \left(\rho_{kl}\right)_{i_1 i_2} \ket{\phi_{1,i_2}} \bra{\phi_{1,i_1}},
    \nonumber \\
    \left(\rho_{kl}\right)_{i_1 i_2}=\frac{1}{N_{kl}^2}\sum_{j=k_1+1}^{k_2} \alpha_{i_1} \beta_{j} a_{i_1 j, kl}^* \alpha_{i_2} \beta_{j} a_{i_2 j, kl} \ .
    \label{eq:reduced density matrix elemnt i1 i2 of relevant state kl for general alpha and beta}
\end{gather}
The entanglement entropy of the partition of $V_1\otimes V_2$ to $V_1$ and $V_2$ reads:
\begin{gather}
    S_{kl}=-\Tr{\rho_{kl} \ln{\rho_{kl}}} \ ,
\end{gather}
where $\rho_{kl}$ is the reduced density matrix (\ref{eq:reduced density matrix elemnt i1 i2 of relevant state kl for general alpha and beta}).
\begin{theorem}
    The probability for the final state to be a product state is always greater than zero.
\end{theorem}

\begin{proof}
    We will use the following Lemma:
    \begin{lemma}
        \label{lemma:If pkk=0 then phikl is a product state}
        If the probability for the kth irrelevant state is zero, $p_{kk}=0$, then every relevant state $(k,l)$ is a product state.
    \end{lemma}

    \begin{proof}
        If $p_{kk}=0$ then from  (\ref{eq:probability of non-relevant state kk for general alpha and beta}) it follows that $\alpha_i U_{ik}=0$ for every $1\leq i \leq k_1$, or $\beta_j U_{jk}=0$ for every $k_1+1\leq j \leq k_1+k_2$. If $\alpha_i U_{ik}=0$ for every $1\leq i \leq k_1$, then by (\ref{eq:aijkl of relevant state kl for general alpha and beta}) $\alpha_ia_{ij,kl}=\alpha_i U_{il} U_{jk}$. Substituting this in (\ref{eq:wave function kl for general alpha and beta}) gives the resulting final relevant state $(k,l)$:
        \begin{gather}
            \ket{\Phi}_{kl}=\frac{1}{N_{kl}}\sum_{i=1, j=k_1+1}^{i=k_1, j=k_1+k_2} \alpha_i \beta_j U_{il} U_{jk} \ket{\phi_{1,i}} \ket{\phi_{2,j}}
            \nonumber \\
            =\frac{1}{N_{kl}} \left(\sum_{i=1}^{i=k_1} \alpha_i U_{il} \ket{\phi_{1,i}} \right) \left(\sum_{j=k_1+1}^{j=k_1+k_2} \beta_j U_{jk} \ket{\phi_{2,j}} \right) \ ,
        \end{gather}
        which is a product state. Using the same argument, the relevant state $(k,l)$ will be a product state in the case where $\beta_j U_{jk}=0$ for every $k_1+1\leq j \leq k_1+k_2$.
    \end{proof}
    Next, suppose that the probability for the final state to be a product state is zero. Thus, the probability for every irrelevant state is zero, $p_{kk}=0$ for every $k$. From the Lemma \ref{lemma:If pkk=0 then phikl is a product state} it follows that every relevant state is a product state, thus the final state will be a product state with probability $1$, leading to a contradiction.
\end{proof}

\if{
{\color{blue}It is worth noting that by the graph in figure 11 in \cite{rimock2024generalized} (attached as figure \ref{fig:PvsS_threshold}), it appears that the probability for a product state can be smaller from any chosen positive number --- because it seems that for every probability $p$ less than one there exist a target entropy with $p=P(S>S_{target})$, which means there exist some unitary $U$ from size $4$ on $4$ with probability $p$ of entanglement entropy not below the target entropy, hence probability at most $1-p$ of a product state, where $1-p$ can be small as we desire.}

\begin{figure}
    \includegraphics[width=1.\linewidth] {fig/plot_experiment_prob_above_th_v_s_optimized_haar_measure.png}
    \caption{a}
    \label{fig:PvsS_threshold}
\end{figure}
}\fi
\subsection{Maximally Entangled Final State is Possible only for $k_1=k_2=2$}
\label{subsec:Achieving maximally entangled final state is impossible unless}
In the following we will analyze the conditions for a specific relevant state $(k,l)$ (\ref{eq:wave function kl for general alpha and beta}) to be maximally entangled, and  
generalize the proof in \cite{Calsamiglia_2002}.

\begin{theorem}
    Unless $k_1,k_2 \leq 2$, the resulting final state of the fusion (\ref{eq:wave function kl for general alpha and beta}) can never be maximally entangled. That is, $\rho_{kl}$ cannot be a scalar matrix, set aside the case where $k_1< k_2$ and $\rho_{kl}$ is a scalar matrix when we trace out $\ket{\phi_{2,j}}$, but not when we trace out $\ket{\phi_{1,i}}$.
\end{theorem}

\begin{proof}
    Assume the state (\ref{eq:wave function kl for general alpha and beta}) to be maximally entangled, thus the reduced density matrix $\rho_{kl}$ (\ref{eq:reduced density matrix elemnt i1 i2 of relevant state kl for general alpha and beta}) is proportional to the identity matrix:
\begin{equation}
    \left(\rho_{kl}\right)_{i_1 i_2} =\frac{1}{k_1} \delta _ {i_1 i_2} \ .
    \label{eq:reduced density matrix of relevant state kl for general alpha and beta equals to identity matrix}
\end{equation}
Define the variables:
\begin{gather}
    X_{k'l'}=\frac{k_1}{N_{kl}^2}\sum_{j=k_1+1}^{j=k_2} |\beta_j |^2 U_{jk'}^{*} U_{jl'} \ ,
    \nonumber \\
    V_{s k'} = \alpha_s U_{s k'} \ ,
    \label{eq:X element and V element of relevant state kl for general alpha and beta}
\end{gather}
with which we can write the reduced density matrix elements $\left(\rho_{kl}\right)_{i_1 i_2}$ (\ref{eq:reduced density matrix elemnt i1 i2 of relevant state kl for general alpha and beta}) as:
\begin{gather}
    k_1 \left(\rho_{kl}\right)_{i_1 i_2}= X_{ll} V_{i_1 k}^* V_{i_2 k} + X_{kk} V_{i_1 l}^* V_{i_2 l}
    \nonumber \\
    + X_{lk} V_{i_1 k}^* V_{i_2 l} + X_{kl} V_{i_1 l}^* V_{i_2 k} \ .
    \label{eq:reduced density matrix elemnt i1 i2 of relevant state kl for general alpha and beta by X and V}
\end{gather}

Define the length two vectors $\{v_i\}_{i=1}^{i=k_1}$:
\begin{gather}
    v_i=
    \begin{bmatrix}
        V_{i k} \\
        V_{i l}
    \end{bmatrix} \ ,
    \label{eq:vector vi for relevant state kl for general alpha and beta}
\end{gather}
and the $2$ by $2$ matrix:
\begin{gather}
    A=
    \begin{bmatrix}
        X_{ll} & X_{lk} \\
        X_{kl} & X_{kk}
    \end{bmatrix}
    =\begin{bmatrix}
        X_{ll} & X_{kl}^* \\
        X_{kl} & X_{kk}
    \end{bmatrix} \ .
    \label{eq:A matrix of relevant state kl for general alpha and beta}
\end{gather}
Equation (\ref{eq:reduced density matrix of relevant state kl for general alpha and beta equals to identity matrix}) becomes:
\begin{gather}
    v_{i_1}^{\dag} A v_{i_2} = \delta_{i_1 i_2} \ .
\end{gather}

If $k_1>2$ then the vectors $\{v_i\}_{i=1}^{i=k_1}$ are linearly dependent because they belong to a two-dimensional vector space. Thus, there exist complex numbers $\{z_i \}_{i=1}^{i=k_1}$ which are not all zero such that
\begin{equation}
    \sum_{i=1}^{i=k_1} z_i v_i = 0 \ .
    \label{eq:linearly dependency of vi}
\end{equation} 
Multiplying (\ref{eq:linearly dependency of vi}) by $v_i^{\dag} A$ gives:
\begin{equation}
    0=v_i^{\dag} A \sum_{i'=1}^{i'=k_1} z_{i'} v_{i'} = \sum_{i'=1}^{i'=k_1} z_{i'} v_i^{\dag} A v_{i'} = \sum_{i'=1}^{i'=k_1} z_{i'} \delta_{i i'} = z_{i} \ ,
\end{equation}
therefore all the $z_i$ are zero, hence a contradiction. Thus, $k_1 \leq 2$. Similarly we prove that $k_2 \leq 2$. 
\end{proof}

\begin{theorem}
    The rank of the reduced density matrix (\ref{eq:reduced density matrix elemnt i1 i2 of relevant state kl for general alpha and beta}) is less or equal to $2$. 
    \label{Theorem:rank(rho) is less than 2}
\end{theorem}

\begin{proof}

Define the $k_1$ length vectors:
   \begin{equation}
        u_k=(V_{1k},V_{2k},...,V_{k_1 k}), \hspace{0.2cm} u_l=(V_{1l},V_{2l},...,V_{k_1 l}) \ .
    \end{equation}
        Denote $D=span\left(u_k,u_l\right)$, which is a subspace of $\mathbb{C}^{k_1}$ with dimension $\leq 2$. Hence, the dimension of the orthogonal complement of $D$, $D^{\perp}$, is $\geq k_1-2$. Choose some vector $z\in D^{\perp}$ and denote its elements as $z_1, z_2, ... z_{k_1}$ that satisfy (\ref{eq:linearly dependency of vi}). Now, compute $\rho_{kl} z$ using (\ref{eq:reduced density matrix elemnt i1 i2 of relevant state kl for general alpha and beta}):
            \begin{gather}
        \left(\rho_{kl}z\right)_{i_1} = \sum_{i_2=1}^{k_1} (\rho_{kl})_{i_1 i_2} z_{i_2} =\sum_{i_2=1}^{k_1} \frac{1}{k_1}( X_{ll} V_{i_1 k}^* V_{i_2 k} 
        \nonumber \\
        + X_{kk} V_{i_1 l}^* V_{i_2 l}
        + X_{lk} V_{i_1 k}^* V_{i_2 l} + X_{kl} V_{i_1 l}^* V_{i_2 k}) z_{i_2} \ ,
        \end{gather}
        which can be recast as:
        \begin{gather}
       \left(\rho_{kl}z\right)_{i_1} =\frac{1}{k_1} \left(X_{ll} V_{i_1 k}^* + X_{kl} V_{i_1 l}^* \right) \sum_{i_2=1}^{i_2=k_1} z_{i_2} V_{i_2 k} 
        \nonumber \\
        +\frac{1}{k_1} \left(X_{kk} V_{i_1 l}^* 
        + X_{lk} V_{i_1 k}^* \right) \sum_{i_2=1}^{i_2=k_1} z_{i_2} V_{i_2 l} = 0 \ .
    \end{gather}
     When $\sum_{i_2=1}^{i_2=k_1} z_{i_2} V_{i_2 k}$ and $\sum_{i_2=1}^{i_2=k_1} z_{i_2} V_{i_2 l}$ vanish as the first and second elements of equation (\ref{eq:linearly dependency of vi}). Thus, every $z\in D^{\perp}$ is also in $\ker \left(\rho_{kl} \right)$, hence $\dim \left( \ker\left(\rho_{kl}\right) \right)\geq \dim \left( D^{\perp} \right) \geq k_1-2$, which means the rank of $\rho_{kl}$ is at most $2$.
\end{proof}

\begin{remark}
    One can define the $2$ by $k_1$ matrix:
    \begin{equation}
        V=\left(v_1 \hspace{0.1cm} v_2 ... v_{k_1}\right) \ .
        \label{eq:V matrix of relevant state kl for general alpha and beta}
    \end{equation}
    Then from (\ref{eq:reduced density matrix elemnt i1 i2 of relevant state kl for general alpha and beta by X and V}) it follows that:
    \begin{gather}
        \rho_{kl}=\frac{1}{k_1} V^{\dag} A V 
        \label{eq:reduced density matrix of relevant state kl for general alpha and beta by A and V} \ ,
    \end{gather}
\end{remark}
which proves that $rank(\rho_{kl}) \leq rank(A) \leq 2$.

\subsection{Measuring More than two Qudits}
\label{subsec:Measuring more than two qudits}
In this subsection we generalize Theorem \ref{Theorem:rank(rho) is less than 2} to a case where we combine together more than two clusters. The case of using ancillae that are at a general product state will be derived as a special case.
Let $X_1,X_2,...,X_M$ be $M$ Hilbert spaces that represent the $M$ clusters/ancilla resources, and let $X_m=V_m\otimes W_m$ when $W_m$ represent the chosen qudits, while $V_m$ represent the remaining cluster $m$ (or in the case of ancilla, the empty set). For every $1\leq m \leq M$ denote $k_m=min(\dim (V_m),\dim (W_m))$. By Schmidt decomposition, for every $1\leq m \leq M$ there exists an orthonormal basis $\{\ket{\phi_{m,i}}\}_{i=1}^{i=dim(V_m)}$ for 
$V_m$, $\{\ket{\psi_{m,i}}\}_{i=1}^{i=dim(W_m)}$ for $W_m$ and real non-negative numbers $\{\alpha_{m,i} \}_{i=1}^{i=k_m}$, such that the quantum state $\ket{\Phi_m}$ of $X_m$ is:
\begin{gather}
    \ket{\Phi_m}=\sum_{i=1}^{k_m} \alpha_{m,i} \ket{\phi_{m,i}} \ket{\psi_{m,i}} \ ,
\end{gather}
and the full quantum state is the product state:
\begin{gather}
    \ket{\Phi}=\prod_{m=1}^{M} \ket{\Phi_m} \ .
    \label{eq:Full Quantum State For More Than Two Qudits And Non Entangled Ancilla}
\end{gather}
If the $m$ channel is an ancilla then:
\begin{gather}
    k_m=1,~~~\ket{\phi_{m,1}}=\ket{vac} \ .
    \label{eq:conditions for ancila}
\end{gather}
We now perform an isometric transformation from the space $W_1\oplus W_2 \oplus ... \oplus W_M$ to a bigger Hilbert space, which we represent as a unitary transformation $U^\dag$ on the appropriate creation operators of $W_1\oplus W_2 \oplus ... \oplus W_M\oplus W_{m+1}$, where $W_{m+1}$ includes the vacuum modes:
\begin{gather}
    \begin{bmatrix}
        a_{1,1}^{\dag} \\
        \vdots \\
        a_{1,k_1}^{\dag} \\
        a_{2,1}^{\dag} \\
        \vdots \\
        a_{2,k_2}^{\dag} \\
        \vdots \\
        a_{M,1}^{\dag} \\
        \vdots \\
        a_{M,k_M}^{\dag} \\
        vac_1^{\dag} \\
        \vdots \\
        vac_{k_{M+1}}^{\dag} 
    \end{bmatrix}
    = U
    \begin{bmatrix}
        c_1 \\
        \vdots \\
        c_{K}
    \end{bmatrix} \ .
    \label{eq:fusion matrix for many qudits}
\end{gather}
We denoted $k_{M+1}=dim(W_{M+1})$ and $K=\sum_{m=1}^{M+1} k_m$. After the transformation, the quantum state becomes:
\begin{gather}
    \ket{\Phi}=\sum_{1\leq l_1 \leq l_2 \leq ... \leq l_M \leq K} (\sum_{i_1=1,i_2=1,...,i-M=1}^{i_1=k_1,i_2=k_2,...,i_M=k_M} ((\prod_{m=1}^{M} 
    \nonumber \\
    \alpha_{m,i_m}) a_{i_1,i_2,...,i_M,l_1,l_2,...,l_M} \prod_{m=1}^{M} (\ket{\phi_{m,i_m}}))(\prod_{m=1}^{M} c_{l_m}^{\dag} \ket{vac})) \ ,
    \label{eq:Total Wave Function By clm}
\end{gather}
where the coefficients $a_{i_1,i_2,...,i_M,l_1,l_2,...,l_M}$ are
\begin{gather}
    a_{i_1,i_2,...,i_M,l_1,l_2,...,l_M}=\sum_{\tau \in S_M} \prod_{m=1}^{M} U_{f(m,i_m),l_{\tau (m)}} \ ,
    \label{eq:a coefficient for non entangled M-2 ancilla}
\end{gather}
and $f(m,i_m)=i_m+\sum_{m'=1}^{m-1}k_{m'}$ for $m>1$ and $f(1,i_1)=i_1$. The set $S_M$ is the group of permutations of ${1,2,...,M}$, except for the case when not all the $l_i$ are different from each other, where we identify different permutations $\tau,\sigma$ for which $(l_{\tau_1},...,l_{\tau_M})=(l_{\sigma_1},...,l_{\sigma_M})$ as the same and taking into account only one of them in the sum.
Note that the state $\prod_{m=1}^{M} c_{l_m}^{\dag} \ket{vac}$ is the state of measuring one photon in each $l_m$ channel except for the case, where some of the $l_m$ are equal, for which there is some normalization factor.

From that it follows that if we measure photons in the channels $l_1,l_2,...,l_M$, the resulting quantum state which belongs to the Hilbert space $\prod_{m=1}^{M} V_m$ is:
\begin{gather}
    \ket{\Phi_{l_1,l_2,...,l_M}}=\frac{1}{N_{l_1,l_2,...,l_M}} \sum_{i_1=1,i_2=1,...,i-M=1}^{i_1=k_1,i_2=k_2,...,i_M=k_M} ((\prod_{m=1}^{M} 
    \alpha_{m,i_m})  \cdot
    \nonumber \\
   \cdot a_{i_1,i_2,...,i_M,l_1,l_2,...,l_M} \prod_{m=1}^{M} (\ket{\phi_{m,i_m}})) \ .
   \label{eq:Resulting wave function for M-2 non entangled ancilla}
\end{gather}

Note that if the $m$ channel is an ancilla, then $k_m=1$, $\alpha_{m,1}=1$ and $\ket{\phi_{m,1}}=\ket{vac}$. For the partition of $\prod_{m=1}^{M} V_m$ to $m$ sub-spaces $V_1,V_2,...,V_M$ the regular Von-Neuman entropy is not defined, and instead one can define a vector of length $M$, with the $i$th element of the vector being the Von-Nueman entropy associated with the reduced density matrix that arises from tracing out all the $V_j$ channels except for $V_i$. We will now compute this reduced density matrix for $i=1$, with the results being relevant for every $i$ with the appropriate changes:
\begin{gather}
    \rho_{l_1,l_2,...,l_M}=\sum_{i_1=1,j_1=1}^{i_1=k_1,j_1=k_1} (\rho_{l_1,l_2,...,l_M})_{i_1,j_1} \ket{\phi_{1,i_1}} \bra{\phi_{1,j_1}}
    \nonumber \\
    (\rho_{l_1,l_2,...,l_M})_{i_1,j_1}=\frac{\alpha_{1,i_1}\alpha_{1,j_1}}{N_{l_1,l_2,...,l_M}^2}\sum_{i_2=1,...,i_M=1}^{i_2=k_2,...,i_M=k_M} 
    \nonumber \\
    (\prod_{m=2}^{M}\alpha_{i_m}^2) a_{i_1,i_2,...,i_M,l_1,l_2,...,l_M} a_{j_1,i_2,...,i_M,l_1,l_2,...,l_M}^* \ .
\end{gather}

Define the variables:
\begin{gather}
   X_{n_2,...,n_M,r_2,...,r_M}=\frac{\alpha_{1,i_1}\alpha_{1,j_1}}{N_{l_1,l_2,...,l_M}^2} \sum_{i_2=1,...,i_M=1}^{i_2=k_2,...,i_M=k_M} 
    \nonumber \\(\prod_{m=2}^{M}\alpha_{i_m}^2
    U_{f(m,i_m),n_m} U_{f(m,i_m),r_m}^*) \ , 
    \label{eq:XvariablesForMProductAncilla} 
\end{gather}
then the reduced density matrix can be written as:
\begin{gather}
    (\rho_{l_1,l_2,...,l_M})_{i_1,j_1} = 
    \nonumber \\
    =\sum_{\tau,\sigma\in S_M} X_{l_{\tau(2)},...,l_{\tau(M)},l_{\sigma(2)},...,l_{\sigma(M)}} U_{i_1,l_{\tau(1)}} U_{j_1,l_{\sigma(1)}}^* \ .
    \label{eq:reduced density matrix elemnt i1 j1 of relevant state l1...lK for general alpha and beta}
\end{gather}
The condition for $\rho_{l_1,l_2,...,l_M}$ to be $\frac{1}{k_1}Id_{k_1\times k_1}$ gives the equations:
\begin{gather}
\label{eq:condition for the reduced density matrix elemnt i1 j1 of relevant state l1...lK for general alpha and beta being a scalar matrix}
    \sum_{\tau,\sigma\in S_M} X_{l_{\tau(2)},...,l_{\tau(M)},l_{\sigma(2)},...,l_{\sigma(M)}} U_{i_1,\tau(1)} U_{j_1,\sigma(1)}^*=\frac{\delta_{i_1,j_1}}{k_1} \ .
\end{gather}

\begin{theorem}
\label{Theorem:rank is less or equal to M}
    The reduced density matrix that follows from tracing out all the channels $l_1,l_2,...,l_M$ except for one has rank less or equal to $M$.
\end{theorem}

\begin{proof}
    For every $1\leq l\leq K$ define the vectors:
    \begin{gather}
        u_l=(U_{1,l},U_{2,l},...,U_{k_1,l})^*  \ .
    \end{gather}
    Denote $D=\spn{\left(u_{l_1},u_{l_2},...,u_{l_M}\right)}$ which is a subspace of $\mathbb{C}^{k_1}$ with dimension $\leq M$. Hence, the dimension of the orthogonal complement of $D$, $D^{\perp}$, is $\geq k_1-M$. Choose some vector $z\in D^{\perp}$ and denote its elements as $z_1, z_2, ... z_{k_1}$ that satisfy:
    \begin{gather}
        \sum_{i_1=1}^{i_1=k_1} z_{i_1}U_{i_1,l_m}^*=\sum_{i_1=1}^{i_1=k_1} z_{i_1}(u_{l_m})_{i_1}=0  \ ,
    \end{gather}
    for every $1\leq m \leq K$.
    Next, compute $\rho_{l_1,l_2,...,l_M} z$ using (\ref{eq:reduced density matrix elemnt i1 j1 of relevant state l1...lK for general alpha and beta}):
    \begin{gather}
        \left(\rho_{l_1,l_2,...,l_M}z\right)_{i_1} = \sum_{j_1=1}^{k_1} (\rho_{l_1,l_2,...,l_M})_{i_1 j_1} z_{j_1} = 
    \nonumber \\
    =\sum_{j_1=1}^{j_1=k_1}z_{j_1}\sum_{\tau,\sigma\in S_M} X_{l_{\tau(2)},...,l_{\tau(M)},l_{\sigma(2)},...,l_{\sigma(M)}} U_{i_1,l_{\tau(1)}} U_{j_1,l_{\sigma(1)}}^* \ ,
\end{gather}
 which can be written as:   
 \begin{gather}
  \left(\rho_{l_1,l_2,...,l_M}z\right)_{i_1}
    =\sum_{\tau,\sigma\in S_M} X_{l_{\tau(2)},...,l_{\tau(M)},l_{\sigma(2)},...,l_{\sigma(M)}} \nonumber\\U_{i_1,l_{\tau(1)}} \left(\sum_{j_1=1}^{j_1=k_1}z_{j_1} U_{j_1,l_{\sigma(1)}}^*\right)
   = 0 \ .
    \end{gather}
   Thus, every $z\in D^{\perp}$ is also in $\ker \left(\rho_{l_1,l_2,...,l_M} \right)$, hence $\dim \left( \ker\left(\rho_{l_1,l_2,...,l_M}\right) \right)\geq \dim \left( D^{\perp} \right) \geq k_1-M$, which means the rank of $\rho_{l_1,l_2,...,l_M}$ is at most $M$.
\end{proof}

By setting all the channels except two , $3,4,...,M$, to be ancilla (by (\ref{eq:conditions for ancila})), we get the case of generalized fusion type II for two d-clusters that are $V_1$ and $V_2$, and regular von-neuman entropy instead of a vector. From this we deduce that in order to fuse to d-clusters by the generalized fusion type-II, we must use at least $d-2$ ancilla qudits. 

\subsection{Measuring More than two Qudits with Entangled Ancilla Qudits}
\label{subsec:Measuring more than two qudits with Entangled Ancilla Qudits}

We now assume that the ancilla qudits, once represented by a wave function in $X_3\otimes...\otimes X_M$, can be entangled, and not necessarily a product state as in equation (\ref{eq:Full Quantum State For More Than Two Qudits And Non Entangled Ancilla}). Instead, assume we have a general Hilbert space $W=W_3\otimes ...\otimes W_M$ on $M$ channels, where $W_i$ populated the appropriate wave function is $\ket{\psi_i}$ --- meaning we set $k_i=1$ for every $i\geq 3$, and we start with a wave function of the ancilla that is the most general form of a wave function in $W$:


{\small \begin{gather}
    \ket{\Phi_{ancilla}}=\sum_{A\subset \{{1,...,M\}},|A|\geq 2,A_1=1,A_2=2} \alpha_A \ket{\psi_{A_3}}...\ket{\psi_{A_{|A|}}} \ , 
    \label{eq:Ancilla state}
\end{gather}}

    \label{eq:Ancilla state}

in which we have terms with any number of photons between 0 and $M-2$ (depending on $|A|-2$). If $A=\{1,2\}$ then the appropriate term in the sum is $\alpha_{\{1,2\}}\ket{vac}$. 
Now, the full quantum state is (as analog to (\ref{eq:Full Quantum State For More Than Two Qudits And Non Entangled Ancilla})):

\begin{gather}
    \ket{\Phi}=\ket{\Phi_1}\ket{\Phi_2}\ket{\Phi_{ancilla}} \ .
\end{gather}

For simplicity, from here on, for every $2\leq M' \leq M$ we denote by $\varpi_{M'}$ the set of all subsets of $\{1,...,M\}$ with size $M$ and $A_1=1,A_2=2$ (where $A_p$ is the p-th element of $A$ where $A$ elements ordered from lower to higher). 
Now, after the transformation (\ref{eq:fusion matrix for many qudits}), we will get the representation:

\begin{gather}
    \ket{\Phi}=\sum_{M'=2}^{M} \sum_{1\leq l_1 \leq ... \leq l_{M'} \leq K} \sum_{A\in \varpi_{M'}} 
    \sum_{i_1=1}^{i_1=k_{A_1}}\sum_{i_2=1}^{i_2=k_{A_2}}
    \nonumber \\
    \alpha_{1,i_1}\alpha_{2,i_2}\alpha_A a_{i_1,i_2,A,l_1,l_2,...,l_{M'}} \ket{\phi_{1,i_1}} \ket{\phi_{2,i_2}}
    (\prod_{m=1}^{M'} c_{l_m}^{\dag} \ket{vac})) \ ,
\end{gather}

where:

\begin{gather}
    a_{i_1,i_2,A,l_1,l_2,...,l_{M'}}=\sum_{\tau\in S_{M'}} \prod_{m=1}^{M'} U_{f(A_m,i_{A_m}),l_{\tau(A_m)}} \ ,
\end{gather}

where $i_{A_m}=1$ for $A_m\geq 3$. Now $M'$ is the number of photons that varies between $2$ and $M$, depending on the size of $A$ --- each specific $A$ term in (\ref{eq:Ancilla state}), contributed to all the options of measuring exactly $|A|$ qudits in the output channels. As in (\ref{eq:a coefficient for non entangled M-2 ancilla}), $S_{M'}$ is the group of permutations of ${1,2,...,M'}$, except for the case when not all the $l_i$ are different from each other, where we identify different permutations $\tau,\sigma$ for which $(l_{\tau_1},...,l_{\tau_M'})=(l_{\sigma_1},...,l_{\sigma_M'})$ as the same and taking into account only one of them in the sum. Now, as in (\ref{eq:Resulting wave function for M-2 non entangled ancilla}), after measuring all the channels $c_1,...,c_K$, if the result is measuring photons in the channels $l_1,...,l_{M'}$ (where $M'\geq 2$), the resulting quantum state which belongs to the Hilbert space $V_1\otimes V_2$ is:

\begin{gather}
    \ket{\Phi_{l_1,l_2,...,l_{M'}}}=\frac{1}{N_{l_1,l_2,...,l_{M'}}} \sum_{A\in \varpi_{M'}}
    \nonumber \\
    \alpha_{1,i_1}\alpha_{2,i_2}\alpha_A a_{i_1,i_2,A,l_1,l_2,...,l_{M'}} \ket{\phi_{1,i_1}} \ket{\phi_{2,i_2}} \ .
\end{gather}




Now, the reduced density matrix achieved by tracing out $V_2$ is:

\begin{gather}
    \rho_{l_1,l_2,...,l_{M'}}=\sum_{i_1=1,j_1=1}^{i_1=k_1,j_1=k_1} (\rho_{l_1,l_2,...,l_{M'}})_{i_1,j_1} \ket{\phi_{1,i_1}} \bra{\phi_{1,j_1}}
    \nonumber \\
    (\rho_{l_1,l_2,...,l_{M'}})_{i_1,j_1}=\frac{\alpha_{1,i_1}\alpha_{1,j_1}}{N_{l_1,l_2,...,l_{M'}}^2}\sum_{A\in \varpi_{M'}}
    \sum_{i_2=1}^{i_2=k_2}
    \nonumber \\
    \alpha_{2,i_2} 
    a_{i_1,i_2,A,l_1,l_2,...,l_{M'}} a_{j_1,i_2,A,l_1,l_2,...,l_{M'}}^* \ .
\end{gather}

As in (\ref{eq:XvariablesForMProductAncilla}), define the variables:

\begin{gather}
    X_{n_2,...,n_{M'},r_2,...,r_{M'}}=\frac{\alpha_{1,i_1}\alpha_{1,j_1}}{N_{l_1,l_2,...,l_{M'}}^2}\sum_{A\in \varpi_{M'}}
    \sum_{i_2=1}^{i_2=k_2} \alpha_{2,i_2}^2 \nonumber \\
    \prod_{m=2}^{m=M'} U_{f(
    A_m,i_{A_m}),n_m} U_{f(A_m,i_{A_m}),r_m}^* \ ,
\end{gather}

so the reduced density matrix can be written as:

\begin{gather}
    (\rho_{l_1,l_2,...,l_{M'}})_{i_1,j_1}=
    \nonumber \\ =\sum_{\tau,\sigma\in S_{M'}} X_{l_{\tau(2)},...,l_{\tau(M')},l_{\sigma(2)},...,l_{\sigma(M')}} U_{i_1,l_{\tau(1)}} U_{j_1,l_{\sigma(1)}}^* \ .
    \label{eq:reduced density matrix elemnt i1 j1 of relevant state l1...lK for general alpha and beta ancilla not product}
\end{gather}

This is the exact form as in (\ref{eq:reduced density matrix elemnt i1 j1 of relevant state l1...lK for general alpha and beta}), with only replacing $M$ by $M'$ (the number of channels containing photons is not $M$ anymore --- this is just the maximum number possible). From here, we can just follow the footsteps of the proof of theorem \ref{Theorem:rank is less or equal to M} while replacing $M$ by $M'$, and get the next theorem.

\begin{theorem}
\label{Theorem:rank is less or equal to M'}
    The reduced density matrix has rank less or equal to the number of detected photons.
\end{theorem}

Finally, one can also use equation \ref{eq:condition for the reduced density matrix elemnt i1 j1 of relevant state l1...lK for general alpha and beta being a scalar matrix} with replacing $M$ by $M'$, to get the condition for $\rho_{l_1,l_2,...,l_{M'}}$ to be $\frac{1}{k_1}Id_{k_1\times k_1}$.

\section{Discussion and Outlook}
\label{sec:discussion and outlook}

We established a structural constraint for generalized type-II fusion of qudit cluster states using passive, number-preserving linear optics with number-resolving detection: for any interferometer and two-click heralding event, the Schmidt rank across the fused cut is upper-bounded by the number of measured systems, implying that a correct $d$-dimensional fusion (rank $d$) is impossible without ancillae and, in particular, requires at least $d\!-\!2$ ancillary qudits. This generalizes familiar qubit no-go results to arbitrary $d$ and to non-Bell projections, and yields a technology-agnostic resource threshold for fusion-based MBQC that reproduces the $d\!=\!2$ ceiling and clarifies why constructive high-$d$ schemes have necessarily invoked additional entanglement or multiple Bell pairs (see, e.g. \cite{rimock2024generalized}).

With the ancilla threshold fixed, a natural next question is: \emph{given $d\!-\!2$ ancillae, what maximal heralded success probability is achievable for generalized type-II fusion?} Existing works provide informative lower bounds \cite{Luo_2019,bharoshigh,Bharos_2025,_st_n_2025}, but a sharp characterization remains open. Two complementary routes appear promising: (i) an \emph{analytical} treatment paralleling the method of \cite{rimock2024generalized} applied to the constraints summarized in Eq.~(49), and (ii) \emph{numerical} optimization over interferometer unitaries, either enforcing maximally entangled outputs as in \cite{schmidt2024generalizedfusionsphotonicquantum} or relaxing to bounded-entanglement targets to trade entanglement for success probability (cf. \cite{rimock2024generalized}). 

While our bounds assume passive, number-preserving optics with number-resolving detection (no feed-forward or nonlinearities), it is important to assess how modest resource relaxations—adaptive feed-forward, weak nonlinearities, or nonclassical ancillae (e.g., squeezed states) modify the accessible reduced-rank structure and potentially reduce ancilla overhead. At the architecture level, our results motivate micro-fusion strategies that allocate a small number of ancillae to maximize rank across graph cuts most relevant for percolation and fault-tolerant encodings, and suggest design curves trading dimension, ancilla count, and success probability. Experimentally, the rank bound is robust to loss and mode mismatch (it cannot be exceeded without additional resources), thereby providing a dependable benchmark for budgeting interferometer depth, detector efficiency, and ancilla supply in high-dimensional, fusion-based photonic quantum computing.

\if{
We have established a structural constraint for generalized type-II fusion of qudit cluster states using passive, number-preserving linear optics with number-resolving detection: for any interferometer and two-click heralding outcome, the Schmidt rank across the fused cut is upper-bounded by the number of measured systems, implying that a correct $d$-dimensional fusion (rank $d$) is impossible without ancillae and, in particular, requires at least $d-2$ ancillary qudits. This generalizes qubit no-go results to arbitrary $d$ and to non-Bell projections, providing a technology-agnostic resource threshold for fusion-based MBQC. Practically, the bound clarifies why prior constructive schemes necessarily invoke additional entanglement or multiple Bell pairs, and it supplies a baseline footprint for high-$d$ fusion constructions. Several directions follow. First, with the ancilla threshold fixed, it is natural to quantify optimal heralding probabilities and entanglement quality achievable by passive interferometers, yielding design curves that trade dimension, ancilla count, and success. Second, modest resource relaxations—adaptive feed-forward, weak nonlinearities, or nonclassical ancillae may alter the accessible reduced-rank structure and reduce overhead; isolating precisely which assumptions our proof uses will sharpen the boundary between impossible and possible fusion. Finally, at the architecture level, our bound motivates micro-fusion strategies that allocate few ancillae to maximize rank across graph cuts relevant for percolation and fault tolerance, while experimental efforts should target low-loss interferometers and efficient number-resolving detectors, for which our constraint remains a robust benchmark.

We defined a generalized type-II fusion of clusters of qudits, and proved that a minimal number of $d-2$ ancilla qudits is required for performing this fusion, as is the case for the standard
type-II fusion for qudits.
The natural question that follows is, given $d-2$ ancilla qudits, what is the maximal success probability of the generalized type-II fusion.
Certain lower bounds have been introduced in \cite{Luo_2019,bharoshigh,Bharos_2025}. A potential analytical avenue to answer the question is to perform an analogous
analysis to that in \cite{rimock2024generalized} for the conditions in (\ref{eq:condition for the reduced density matrix elemnt i1 j1 of relevant state l1...lK for general alpha and beta being a scalar matrix}). Another direction is to perform a numerical analysis, by either focusing on obtaining a maximally-entangled state as in \cite{schmidt2024generalizedfusionsphotonicquantum}, or lowering the required entanglement entropy, while increasing the probability of success as shown in \cite{rimock2024generalized}.

}\fi
\vspace{0.5cm}

\textbf{Acknowledgments}
We would like to thank Khen Cohen for a valuable discussion.  

\textbf{Funding} - This work is supported in part by the Israeli Science Foundation Excellence Center, the US-Israel Binational Science Foundation, and the Israel Ministry of Science.
\bibliographystyle{unsrt.bst}
\bibliography{main.bib}

@article{BrowneRudolph,
  author    = {D. E. Browne and T. Rudolph},
  title     = {Resource-efficient linear optical quantum computation},
  journal   = {ArXiv},
  year      = {2004},
  note      = {\url{https://doi.org/10.1103/PhysRevLett.95.010501}},
}

@article{ThreePhoton,
  author    = {Gimeno-Segovia, M. and Shadbolt, P. and Browne, D. E. and Rudolph, T.},
  title     = {From three-photon GHZ states to ballistic universal quantum computation},
  year      = {2015},
  note      = {\url{https://doi.org/10.1103/PhysRevLett.115.020502}},
}

@article{FusionBasedQC,
  author    = {Bartolucci, S. and Birchall, P. and Bombín, H. and Cable, H. and Dawson, C. and Gimeno-Segovia, M. and Johnston, E. and Kieling, K. and Nickerson, N. and Pant, M. and Pastawski, F. and Rudolph, T. and Sparrow, C.},
  title     = {Fusion-based quantum computation},
  year      = {2021},
  note      = {\url{https://doi.org/10.1038/s41467-023-36493-1}},
}

@article{CreatingU0,
  author    = {Bernstein, H. J. and Bertani, P.},
  title     = {Experimental Realization of Any Discrete Unitary Operator},
  journal   = {Arxiv},
  year      = {1994},
  note      = {\url{https://doi.org/10.1103/PhysRevLett.73.58}},
}

@article{rimock2024generalized,
  title={Generalized Type II Fusion of Cluster States},
  author={Rimock, Noam and Cohen, Khen and Oz, Yaron},
  journal={arXiv preprint arXiv:2406.15666},
  year={2024}
}

@misc{bartolucci2021creationentangledphotonicstates,
      title={Creation of Entangled Photonic States Using Linear Optics}, 
      author={Sara Bartolucci and Patrick M. Birchall and Mercedes Gimeno-Segovia and Eric Johnston and Konrad Kieling and Mihir Pant and Terry Rudolph and Jake Smith and Chris Sparrow and Mihai D. Vidrighin},
      year={2021},
      eprint={2106.13825},
      archivePrefix={arXiv},
      primaryClass={quant-ph},
      url={https://arxiv.org/abs/2106.13825}, 
}

@misc{schmidt2024generalizedfusionsphotonicquantum,
      title={Generalized fusions of photonic quantum states using static linear optics}, 
      author={Frank Schmidt and Peter van Loock},
      year={2024},
      eprint={2410.20261},
      archivePrefix={arXiv},
      primaryClass={quant-ph},
      url={https://arxiv.org/abs/2410.20261}, 
}

@article{Wang_2020,
   title={Qudits and High-Dimensional Quantum Computing},
   volume={8},
   ISSN={2296-424X},
   url={http://dx.doi.org/10.3389/fphy.2020.589504},
   DOI={10.3389/fphy.2020.589504},
   journal={Frontiers in Physics},
   publisher={Frontiers Media SA},
   author={Wang, Yuchen and Hu, Zixuan and Sanders, Barry C. and Kais, Sabre},
   year={2020},
   month=nov }

@article{Zhou_2003,
   title={Quantum computation based ond-level cluster state},
   volume={68},
   ISSN={1094-1622},
   url={http://dx.doi.org/10.1103/PhysRevA.68.062303},
   DOI={10.1103/physreva.68.062303},
   number={6},
   journal={Physical Review A},
   publisher={American Physical Society (APS)},
   author={Zhou, D. L. and Zeng, B. and Xu, Z. and Sun, C. P.},
   year={2003},
   month=dec }

@article{Browne_2005,
   title={Resource-Efficient Linear Optical Quantum Computation},
   volume={95},
   ISSN={1079-7114},
   url={http://dx.doi.org/10.1103/PhysRevLett.95.010501},
   DOI={10.1103/physrevlett.95.010501},
   number={1},
   journal={Physical Review Letters},
   publisher={American Physical Society (APS)},
   author={Browne, Daniel E. and Rudolph, Terry},
   year={2005},
   month=jun }

@article{Bharos_2025,
   title={Efficient High-Dimensional Entangled State Analyzer with Linear Optics},
   volume={9},
   ISSN={2521-327X},
   url={http://dx.doi.org/10.22331/q-2025-04-18-1711},
   DOI={10.22331/q-2025-04-18-1711},
   journal={Quantum},
   publisher={Verein zur Forderung des Open Access Publizierens in den Quantenwissenschaften},
   author={Bharos, Niv and Markovich, Liubov and Borregaard, Johannes},
   year={2025},
   month=apr, pages={1711} }

@article{bharoshigh,
  title={High-Dimensional Entanglement Generation with Linear Optics},
  author={Bharos, Niv}
}

@article{Luo_2019,
   title={Quantum Teleportation in High Dimensions},
   volume={123},
   ISSN={1079-7114},
   url={http://dx.doi.org/10.1103/PhysRevLett.123.070505},
   DOI={10.1103/physrevlett.123.070505},
   number={7},
   journal={Physical Review Letters},
   publisher={American Physical Society (APS)},
   author={Luo, Yi-Han and Zhong, Han-Sen and Erhard, Manuel and Wang, Xi-Lin and Peng, Li-Chao and Krenn, Mario and Jiang, Xiao and Li, Li and Liu, Nai-Le and Lu, Chao-Yang and Zeilinger, Anton and Pan, Jian-Wei},
   year={2019},
   month=aug }

@article{Calsamiglia_2001,
   title={Maximum efficiency of a linear-optical Bell-state analyzer},
   volume={72},
   ISSN={1432-0649},
   url={http://dx.doi.org/10.1007/s003400000484},
   DOI={10.1007/s003400000484},
   number={1},
   journal={Applied Physics B},
   publisher={Springer Science and Business Media LLC},
   author={Calsamiglia, J. and Lütkenhaus, N.},
   year={2001},
   month=jan, pages={67–71} }

@article{Calsamiglia_2002,
   title={Generalized measurements by linear elements},
   volume={65},
   ISSN={1094-1622},
   url={http://dx.doi.org/10.1103/PhysRevA.65.030301},
   DOI={10.1103/physreva.65.030301},
   number={3},
   journal={Physical Review A},
   publisher={American Physical Society (APS)},
   author={Calsamiglia, John},
   year={2002},
   month=feb }

@article{grice2011arbitrarily,
  title={Arbitrarily complete Bell-state measurement using only linear optical elements},
  author={Grice, Warren P},
  journal={Physical Review A—Atomic, Molecular, and Optical Physics},
  volume={84},
  number={4},
  pages={042331},
  year={2011},
  publisher={APS}
}

@article{Ewert_2014,
   title={Efficient Bell Measurement with Passive Linear Optics and Unentangled Ancillae},
   volume={113},
   ISSN={1079-7114},
   url={http://dx.doi.org/10.1103/PhysRevLett.113.140403},
   DOI={10.1103/physrevlett.113.140403},
   number={14},
   journal={Physical Review Letters},
   publisher={American Physical Society (APS)},
   author={Ewert, Fabian and van Loock, Peter},
   year={2014},
   month=sep }

@article{Kilmer_2019,
   title={Boosting linear-optical Bell measurement success probability with predetection squeezing and imperfect photon-number-resolving detectors},
   volume={99},
   ISSN={2469-9934},
   url={http://dx.doi.org/10.1103/PhysRevA.99.032302},
   DOI={10.1103/physreva.99.032302},
   number={3},
   journal={Physical Review A},
   publisher={American Physical Society (APS)},
   author={Kilmer, Thomas and Guha, Saikat},
   year={2019},
   month=mar }

@article{Bayerbach_2023,
   title={Bell-state measurement exceeding 50% success probability with linear optics},
   volume={9},
   ISSN={2375-2548},
   url={http://dx.doi.org/10.1126/sciadv.adf4080},
   DOI={10.1126/sciadv.adf4080},
   number={32},
   journal={Science Advances},
   publisher={American Association for the Advancement of Science (AAAS)},
   author={Bayerbach, Matthias J. and D’Aurelio, Simone E. and van Loock, Peter and Barz, Stefanie},
   year={2023},
   month=aug }

@article{_st_n_2025,
   title={Fusion for high-dimensional linear-optical quantum computing with improved success probability},
   volume={24},
   ISSN={2331-7019},
   url={http://dx.doi.org/10.1103/l7bg-hc8c},
   DOI={10.1103/l7bg-hc8c},
   number={4},
   journal={Physical Review Applied},
   publisher={American Physical Society (APS)},
   author={Üstün, Gözde and Rieffel, Eleanor G. and Devitt, Simon J. and Saied, Jason},
   year={2025},
   month=oct }

\end{document}